\documentclass[journal,letterpaper,twocolumn,twoside,nofonttune]{IEEEtran}
\usepackage{times}

\usepackage{amsmath}
\usepackage{amsfonts}
\usepackage{bbding}   
\usepackage{stmaryrd}
\usepackage{latexsym}
\usepackage{amssymb}
\usepackage{algorithm2e}
\usepackage{graphics}
\usepackage{graphicx}
\usepackage{subfigure}
\usepackage{colordvi}

\usepackage{mathrsfs} % Define \mathscr, a script font

\usepackage{cite}  
\usepackage{upref}

\usepackage{theorem}
\title{\Huge$\,$\\[-2.75ex]
Algebraic List-decoding of Subspace Codes\\[0.50ex]}

\author{%
\authorblockN{\large{\bf%
Hessam Mahdavifar and Alexander Vardy}\vspace{0.25ex}}\\
\authorblockA{\large
University of California San Diego, La Jolla, CA\,92093, USA\\ 
{\tt \{hessam@ucsd.edu,\,avardy@ucsd.edu\}}}\vspace*{-1.5ex}
}
\renewcommand{\markboth}[2]
{\renewcommand{\leftmark}{#1}\renewcommand{\rightmark}{#2}}

\theoremstyle{plain} 
\theorembodyfont{\normalfont\slshape}

\newtheorem{thm}{Theorem\hspace{-1pt}} 
\newenvironment{theorem}
{\begin{thm}\hspace*{-1ex}{\bf.}}{\end{thm}}

\newtheorem{lem}[thm]{Lemma\hspace{-.75pt}}
\newenvironment{lemma}{\begin{lem}\hspace*{-1ex}{\bf.}}{\end{lem}}

\newtheorem{prop}[thm]{Proposition$\!$}

\newtheorem{cor}[thm]{Corollary$\!$}
\newenvironment{corollary}{\begin{cor}\hspace*{-1ex}{\bf.}}{\end{cor}}

\newtheorem{defn}{Definition$\!$}

\newcounter{enumrom}
\renewcommand{\theenumrom}{(\roman{enumrom})}

\makeatletter
\renewcommand{\@endtheorem}{\endtrivlist}
\makeatother

\makeatletter
\renewcommand{\thefigure}{{\bf \@arabic\c@figure}}
\renewcommand{\fnum@figure}{{\bf Figure}\,\thefigure}
\makeatother

\renewcommand{\QEDclosed}{\mbox{\rule[-1pt]{1.3ex}{1.3ex}}} % 

\newcommand{\cC}{{\cal C}} 

\newcommand{\cE}{{\cal E}} 

\newcommand{\cG}{{\cal G}}

\newcommand{\cP}{{\cal P}}

\DeclareMathAlphabet{\mathbfsl}{OT1}{ppl}{b}{it} %{OT1}{cmr}{bx}{it}

\newcommand{\bu}{\mathbfsl{u}}

\newcommand{\be}[1]{\begin{equation}\label{#1}}
\newcommand{\ee}{\end{equation}} 
\newcommand{\eq}[1]{(\ref{#1})}

\renewcommand{\leq}{\leqslant}
 
\renewcommand{\geq}{\geqslant}

\newcommand{\script}[1]{{\mathscr #1}}

\renewcommand{\Bbb}{\mathbb}
 
\newcommand{\N}{{\Bbb N}}

\newcommand{\Tref}[1]{Theo\-rem\,\ref{#1}}

\newcommand{\Lref}[1]{Lem\-ma\,\ref{#1}}
\newcommand{\Cref}[1]{Co\-ro\-lla\-ry\,\ref{#1}}

\newcommand{\F}{\Bbb{F}}
\newcommand{\Fq}{\Bbb{F}_{\!q}}

\newcommand{\Fqm}{\Bbb{F}_{q^m}}

\newcommand{\fbu}{f_{\bu}}
\newcommand{\fbutwo}{\fbu^{\otimes 2}}

\newcommand{\ftwo}{f^{\otimes 2}}
\newcommand{\fl}{\fbu^{\otimes L}}
\renewcommand{\L}{\script{L}}

\newcommand{\deff}{\mbox{$\stackrel{\rm def}{=}$}}

\newcommand{\al}{\alpha}

\newcommand{\Rs}{R^{*}}

\begin{document}

\maketitle
\thispagestyle{empty}

\begin{abstract}
Subspace codes were introduced in order to correct errors 
and erasures for randomized network coding, in the case 
where network topology is unknown (the noncoherent case). 
Subspace codes are indeed collections of subspaces of a 
certain vector space over a finite field. The Koetter-Kschischang 
construction of subspace codes are similar to Reed-Solomon codes 
in that codewords are obtained by evaluating certain (linearized) polynomials. 
%in a set of points.

\looseness=-1
\hspace*{-1pt}In this paper, we consider the problem of list-decoding 
the~Koet\-ter-Kschischang subspace codes. In a sense,
we are able to achieve for these codes what Sudan was able
to achieve~for Reed-Solomon codes. In order to do so, we have to 
modify and generalize the original Koetter-Kschischang
construction in many important respects. The end result is
this: for any integer $L$, our list-$L$ decoder guarantees
successful recovery of the message subspace provided that the 
normalized dimension of the error is at most
$$
L \,-\, \frac{L(L+1)}{2}\Rs
$$
where $\Rs$ is the \emph{normalized packet rate}. Just as in the
case~of Sudan's list-decoding algorithm, this exceeds the 
previously~best-known error-correction radius $1-\Rs$,
demonstrated by Koetter and Kschischang, for low rates $\Rs$.
\end{abstract}

\begin{keywords} 
list-decoding, subspace codes, operator channel, linearized polynomial, KK construction
\vspace{-1.00ex}
\end{keywords}

%=======================================================================%
%                                                                       %
%    1. INTRODUCTION                                                    %
%                                                                       %
%=======================================================================%
\section{Introduction} 
\label{sec:Introduction}
%\vspace{-.25ex}
\noindent\looseness=-1
\PARstart{T}{he} technique of list-decoding has been used to decode beyond the error-correction diameter bound \cite{E}, \cite{E2} and \cite{GS}. In general, the decoding problem is the problem of finding a codeword which is within a particular distance from a received word. In classical decoding algorithms, the decoding radius is such that decoding spheres around codewords do not intersect which results in diameter bound on the decoding radius. In list-decoding, the receiver reconstructs a list of all possible codewords within a particular distance of the received word. This offers a potential for recovery from errors beyond the traditional error-correction bound.  

In this paper we consider the problem of list-decoding of subspace codes. Although the idea is similar, the nature of problem is different from the classical case in many respects. In \cite{KK}, subspace codes were introduced in the context of \emph{noncoherent} transmission model for random network coding. In noncoherent transmission model, neither transmitter nor receiver is assumed to have any knowledge about the underlying network topology and the particular network coding operations performed at the network nodes. 

Random network coding is a very powerful tool for information transmission in networks\cite{CWJ}, \cite{HKMKE} and \cite{HMKKESL}. In random network coding communication between transmitter and receiver is done in a series of \emph{generations}. During each generation the transmitter transmits a group of packets with a fixed-length through the network. These packets can be regarded as vectors of length $n$ over a finite field $\Fq$. They pass through intermediate nodes of the network. Each intermediate node creates a random $\Fq$-linear combination of the packets that are already available at its input links and sends it through an output link. This is done separately for all of its output links. In this model, we suppose that a number of erroneous packets may be injected into the network. Finally, the receiver collects a number of such generated packets and tries to reconstruct the set of packets injected into the network. The authors of \cite{KK} are led to consider information transmission by the choice of vector space spanned by the set of packets at the transmitter. Intuitively this is the only thing that is preserved through transmission as linear combinations are assumed to be random and unknown to the transmitter and the receiver.

In \cite{KK} the operator channel is defined in order to capture the essence of random network coding model. The input and output of an operator channel are subspaces of a certain vector space called the ambient space. Deletion of vectors from the transmitted subspace is called erasure and addition of vectors to the transmitted subspace is called error. The input is affected by some erasures and errors through the channel and the goal is to recover the input from the output. A Reed-Solomon-like subspace code capable of correcting errors and erasures on the operator channel is introduced in \cite{KK}. We call it KK code (Koetter-Kschischang code) throughout this paper. In KK construction, codewords are obtained by evaluating certain (linearized) polynomials. In Section \ref{sec:two} we briefly review the KK construction, its encoding and decoding. 

What we are doing in this paper, in some sense, is analogous to what Sudan did for list-decoding of Reed-Solomon codes \cite{S}. The main obstacle here is that the ring of linearized polynomials is not commutative. Indeed equations over this ring may have more roots than their degrees. Therefore, the straightforward generalization of list-$1$ decoding may result in an exponential list size. We modify the KK construction in many important respects in order to enable list-decoding. The idea is to consider a commutative subring of the ring of linearized polynomials. However, this causes a rate reduction. We solve this problem by using the normal basis for an extension field of a finite field. In Section \ref{sec:three} this technique is explained for the simplest case, one dimensional codes with list size 2, wherein the KK construction can not correct any error for any rate but we are able to correct one error for $0<\Rs<\frac{1}{3}$, where $\Rs$ is the \emph{packet rate} of the code. The packet rate of a subspace code is simply the number of information packets normalized by the number of encoded packets. This is defined more precisely in Section \ref{sec:two}. In this paper, we use this notion of rate in order to express our results in a more convenient way. The results of Section \ref{sec:three} is generalized to arbitrary list size yet for one dimensional codes in Section \ref{sec:four}. The transmitted message is recovered as long as the dimension of error does not exceed $L-\frac{L(L+1)}{2}\Rs$, where $L$ is the list size. 

Our construction can not be immediately generalized to the case of dimension greater than one. The problem is that we already use the whole space as the root space of the equation from which we extract the message polynomial. In fact increasing the dimension of the code does not give more information at the receiver which results in rate reduction. Therefore, we extend the space root of the interpolation polynomial while the symbols are still from the smaller field. In Section \ref{sec:five} we use this idea in order to construct codes of any dimension. Then we get the normalized decoding radius $L-\frac{1}{2}L(L+1)\Rs$, where $L$ is the list size and $n$ is the dimension of the code. This is similar to the result in the one dimension case. 

We close the paper by discussing some directions for future work and drawing conclusion. 

\vspace{1ex}
%=======================================================================%
%                                                                       %
%    2.Prior Work                 %              
%                                                                       %
%=======================================================================%
\section{Prior Work}
\vspace{.25ex}
\label{sec:two}
In this section, following \cite{KK} we review some required background and some prior work on subspace codes. We explain the operator channel model, the ring of linearized polynomials and the Koetter-Kschischang construction of subspace codes. 

The authors of \cite{KK} introduced the operator channel model in order to capture the essence of random network coding. This is formulated for the case of single unicast i.e. communication between a single transmitter and a single receiver. Let $W$ be a fixed $N$-dimensional vector space over $\Fq$. Then all transmitted and received packets are elements of $W$. Let $\cG(W)$ denote the set of all subspaces of $W$ which is often called the projective geometry of $W$. Let also $\cG(W,n)$ denote the set of all subspaces of $W$ of dimension $n$. For any $V \in \cG(W)$, $\dim(V)$ denotes the dimension of $V$. As $\cG(W)$ is the code alphabet, a metric on $\cG(W)$ is defined as follows. Let $\mathbb{Z}_{+}$ denote the set of non-negative integers. Then the function $d\ :\ \cG(W) \times \cG(W)\ \rightarrow\ \mathbb{Z}_{+}$ is defined as follows:
$$
d(A,B)\hspace{3mm}\deff\ \hspace{3mm} \dim(A+B)-\dim(A\cap B)
$$
It is shown in Lemma 1 in \cite{KK} that the function $d$ is indeed a metric.
\\\textbf{Definition\,1.}\cite{KK} An operator channel $C$ associated with the ambient space $W$ is a channel with input and output alphabet $\cG(W)$. Let $V$ and $U$ denote the input and output of the channel respectively. Then
$$
U=\mathcal{H}_k(V)\oplus E,
$$
where $\mathcal{H}_k$ is an erasure operator which projects $V$ onto a randomly chosen $k$-dimensional subspace of $V$ if $\dim(V)>k$; otherwise, $\mathcal{H}_k$ leaves $V$ unchanged. Also, $E \in \cG(W)$ is an arbitrary error space and $\oplus$ denote the direct sum. The number of erasures is $\rho=\dim(V)-\dim(\mathcal{H}_k(V))$. The number of errors is $t=\dim{E}$.
\hfill\raisebox{-0.5ex}{$\Box$}\vspace{1.0ex}

A subspace code $\cC$ for an operator channel with ambient space $W$ is a non-empty subset of $\cG(W)$. A codeword is an element of $\cC$ which is indeed a subspace of $W$. The minimum distance of $\cC$ is denoted by
$$
D(\cC) \hspace{3mm}\deff\ \hspace{3mm} \min_{X,Y\in \cC:X\neq Y} d(X,Y)
$$
It is proved in \cite{KK} that if
\be{eq21}
2(t+\rho) < D(\cC)
\ee
then a minimum distance decoder for $\cC$ will recover the transmitted subspace $V$ from the received subspace $U$. Conversely if \eq{eq21} is not satisfied, then the minimum distance decoder may fail.
\\\textbf{Definition\,2.}\cite{KK} Let $\mathcal{C}$ be a code associated with the ambient space $W$ of dimension $N$ over $\mathbb{F}_q$. Suppose that the dimension of any $V \in \mathcal{C}$ is at most $n$. Then the rate of the codes $R$ is defined as follows:
\be{symbol-rate}
R\hspace{6pt} {\deff} \hspace{6pt} \frac{\log_q\left|\mathcal{C}\right|}{nN}
\ee
In this paper, we define a new parameter, called the \emph{packet rate} of the code. The packet rate $\Rs$ is defined as follows:
\be{packet-rate}
\Rs \hspace{6pt} {\deff} \hspace{6pt} \frac{\log_{q^m}\left|\mathcal{C}\right|}{n} = 
\frac{\log_q\left|\mathcal{C}\right|}{nm}
\ee
where $q^m$ is the size of the underlying extension field. 
\hfill\raisebox{-0.5ex}{$\Box$}\vspace{1.0ex}

In fact, the rate $R$ of the code is equal to the number of $q$-ary information symbols normalized by the number of $q$-ary symbols injected into the network. This can be interpreted as the \emph{symbol rate} of the code. $\Rs$ is equal to the number of information packets normalized by the number of encoded packets injected into the network. Therefore, we call it the packet rate of the code.

Koetter-Kschischang construction of subspace codes are obtained by evaluating linearized polynomials over a certain set of linearly independent elements of an ambient space $W$. Next, we turn to briefly review linearized polynomials, their main properties and how to define a ring structure on them. Let $\Fq$ be a finite field and let $\F=\Fqm$ be an extension field. Recall from \cite[Ch. 4.9]{MS} that a polynomial $f(X)$ is called a linearized polynomial over $\F$ if it has the form
$$
f(X)=\sum^{s}_{i=0} a_i X^{q^i}
$$  
where $a_i \in \F$, for $i=0,1,\dots,s$. When $q$ is fixed under discussion, we will let $X^{[i]}$ denote $X^{q^i}$. We use the term $q$-degree instead of degree for linearized polynomials. For instance, assuming that $a_s \neq 0$, the linearized polynomial $f(X)$ has $q$-degree $s$ which means that its actual degree is equal to $q^s$.

The main property of linearized polynomials from which they receive their name is the following. Let $f(X)$ be a linearized polynomial over $\F$ and let $\mathbb{K}$ be an extension of $\F$. Then the map taking $\alpha \in \mathbb{K}$ to $f(\alpha) \in \mathbb{K}$ is linear with respect to $\Fq$, i.e. for all $\alpha_1,\alpha_2 \in \mathbb{K}$ and all $\lambda_1,\lambda_2 \in \Fq$,
$$
f(\lambda_1 \alpha_1+\lambda_2 \alpha_2)= \lambda_1 f(\alpha_1)+\lambda_2 f(\alpha_2)
$$
Also, it is proved in \cite{KK} that if two linearized polynomials of $q$-degree at most $k-1$ agrees on at least $k$ linearly independent points, then the two polynomials are identical. 

Addition of two linearized polynomials, $f_1(X)$ and $f_2(X)$, is also a linearized polynomial. However, the multiplication $f_1(X)f_2(X)$ is not necessarily a linearized polynomial. Therefore, in order to have a ring structure the operation $f_1(X) \otimes f_2(X)$ is defined to be the composition $f_1(f_2(X))$ which is always a linearized polynomial. In fact, if $f_1(X)=\sum_{i\geq 0} a_i X^{[i]}$ and $f_2(X)=\sum_{j\geq 0} b_j X^{[j]}$, then 
\be{eq20}
f_1(X) \otimes f_2(X)=f_1(f_2(X))=\sum_{k\geq 0} c_k X^{[k]},\\
\ee
where $\ c_k=\sum^{k}_{i=0} a_i b_{k-i}^{[i]}$. It should be noted that this operation is not commutative. It is easy to construct examples for $f_1(X)$ and $f_2(X)$ such that $f_1(X) \otimes f_2(X)$ is not equal to $f_2(X) \otimes f_1(X)$. 

The set of linearized polynomials over $\Fqm$ forms a non-commutative ring with identity under addition $+$ and composition $\otimes$ and is denoted by $\L_{q^m}[X]$. Though not commutative, the ring of linearized polynomials has many of the properties of a Euclidean domain. In fact, there are two division algorithms: a left division and a right division, i.e. given any two linearized polynomials $f_1(X)$ and $f_2(X)$, there exist unique linearized polynomials $q_L(X)$, $q_R(X)$, $r_L(X)$ and $r_R(X)$ such that
$$
f_1(X)=q_L(X)\otimes f_2(X) + r_L(X) = f_2(X) \otimes q_R(X) + r_R(X),
$$
where $r_L(X)=0$ or $\deg(r_L(X)) < \deg(f_2(X))$ and similarly where $r_R(X)=0$ or $\deg(r_R(X)) < \deg(f_2(X))$. A straightforward modification of polynomial division algorithm can be invoked in order to do left division and right division for linearized polynomials.

Now, we turn to briefly review the encoding and decoding of KK construction. Let $\Fq$ be a finite field, and let $\F=\Fqm$ be an extension field of $\Fq$. The number of information symbols $k$ and the dimension of code $n$ are also fixed. Notice that $\F$ can be regarded as a vector space of dimension $m$ over $\Fq$. Let $A=\left\{\alpha_1,\dots,\alpha_n\right\}$ be a set of $n$ linearly independent vectors in this vector space. 
\\\textbf{Koetter-Kschischang Encoding:}
\\The input to the encoder is a vector $\bu=(u_0,\dots,u_{k-1})$ which consists of $k$ message symbols in $\F$. The corresponding message polynomial is $\fbu(X)=\sum^{k-1}_{i=0} u_iX^{\left[i\right]}$. Then the corresponding codeword $V$ is the $\Fq$-linear span of the set $\left\{(\alpha_i,f(\alpha_i)):1 \leq i \leq n\right\}$.

The code $\cC$ is the set of all possible codeword $V$. The ambinet space $W$ is indeed equal to $\left\langle A\right\rangle \oplus \F = \left\{(\alpha,\beta):\alpha\in\left\langle A\right\rangle, \beta \in \F\right\}$ which has dimension $n+m$ over $\Fq$.

Suppose that $V$ is transmitted over the operator channel and a subspace $U$ of $W$ of dimension $r$ is received. 
\\\textbf{Koetter-Kschischang Decoding:}
\\Let $(x_i,y_i), i=1,2,\dots,r$ be a basis for $U$. Construct a nonzero bivariate polynomial $Q(X,Y)$ of the form
$$
Q(X,Y)=Q_0(X)+Q_1(Y),\ 
$$
where $Q_0$ and $Q_1$ are linearized polynomials over $\F$, $Q_0$ has $q$-degree at most $\omega-1$ and $Q_1$ has $q$-degree at most $\omega-k$ such that
$$
Q(x_i,y_i)=0\ \text{for}\ i=1,2,\dots,r 
$$
The parameter $\omega$ will be specified later. Then solve the equation $Q(X,f(X))=0$ for $f(X)$ to recover the message polynomial.

Suppose that $r=n-\rho+t$, where $\rho$ is the number of erasures and $t$ is the number of errors. It is proved in \cite{KK} that if $\rho+t < n-k+1=\frac{D(\cC)}{2}$ which is the necessary and sufficient condition for minimum distance decoder as in \eq{eq21}, then one can choose 
$$
\omega = \left\lceil \frac{r+k}{2}\right\rceil
$$ 
and the decoding algorithm successfully recovers the transmitted message. Therefore, the normalized decoding radius $\tau_{\text{KK}}$ is given as follows:
\be{tau-KK}
\tau_{\text{KK}} = \frac{n-k+1}{n} 
\ee
The packet rate $\Rs$ of the Koetter-Kschischang code is:
$$
\Rs=\frac{\text{log}_{q^m}(\left|\cC\right|)}{n}=\frac{k}{n}
$$
Therfore, $\tau_{\text{KK}}$ is approximately equal to $1-\Rs$.

\vspace{1ex}
%=======================================================================%
%                                                                       %
%    The Simplest List-decoding                                             %              
%                                                                       %
%=======================================================================%
\section{The Simplest List-decoding}
\vspace{.25ex}
\label{sec:three}

We start this section with a brief review of Sudan's list-decoding algorithm of Reed-Solomon codes . Then we justify why it is necessary to modify KK construction in order to enable list-decoding. As the first attempt for list-$2$ decoding, a simple generalization of KK construction is proposed in Section \ref{sec:three-A}. However, we shall see that a list of size $2$ can not be guaranteed as a result of the ring of linearized polynomials being non-commutative. Therefore, we further modify the construction to solve this problem in Section \ref{sec:three-B}. However, this modification results in a rate reduction by a factor of $m$. To compensate this reduction we exploit properties of a normal basis of $\Fqm$ over $\Fq$ in Section \ref{sec:three-C}. Having set all that we explain the encoding and list-decoding of this new construction of subspace codes in Section \ref{sec:three-D}. At the end, the parameters of the code are discussed in Section \ref{sec:three-E}.

First, we briefly review Sudan's list-decoding algorithm of Reed-Solomon codes \cite{S}. The construction of Reed-Solomon codes is as follows. Let $\Fq$ be a finite field. The parameters $k$, the number of information symbols, and $n$, the length of the code are fixed and $k \leq n \leq q-1$. The message is a vector $\bu=(u_0,u_1,\dots,u_{k-1})$ consisting of $k$ information symbols over $\Fq$. The corresponding codeword is $(\fbu(\al_1),\fbu(\al_2),\dots,\fbu(\al_n))$, where $\fbu (X) = \sum^{k-1}_{i=0} u_i X^i$ is the message polynomial and $\al_1,\al_2,\dots,\al_n$ are $n$ distinct and fixed elements of $\Fq$. This codeword is transmitted through the channel. Given the channel output $(y_1,y_2,\dots,y_n)$, Sudan's list-decoding algorithm constructs the bivariate interpolation polynomial
$$
Q(X,Y) = Q_0(X)+Q_1(X)Y+\dots+Q_L(X)Y^L
$$
such that $Q(\al_i,y_i)=0$ for all i, subject to certain degree constraints. Then if not too many errors 
have occurred, $Q(X,\fbu(X)) \equiv 0$, and the message can be recovered by finding all the factors (at most $L$ of them) of $Q(X,Y)$ of the form $Y-F(X)$.

Now, we return to the construction of subspace codes. Let $\Fq$ be a finite field and $\F=\Fqm$ be an extension field of $\Fq$. For ease of notation, let $f^{\otimes L}(X)$ denote the composition of $f(X)$ with itself $L$ times for any linearized polynomial $f(X)$. Indeed $f^{\otimes 1}(X)=f(X)$. Also, we define $f^{\otimes 0}(X)$ to be equal to $X$. Same as in the KK construction, $A=\left\{\alpha_1,\dots,\alpha_n\right\} \subset \Fqm$ is a fixed set of $n$ linearly independent vectors over $\Fq$. The first step in modifying the KK construction in order to enable list-$2$ decoding is the following. We transmit $\fbutwo (\al_i)$ along with $\al_i$ and $\fbu(\al_i)$, where $\fbu$ is the message polynomial.  This is one important difference between this work and Sudan list-decoding algorithm of RS codes. In Sudan's algorithm, there is no need to modify the Reed-Solomon code. One can compute powers of the received symbols $y_i$ at the decoder. In fact,  once $y_i$ is given, all powers of $y_i$ come for free whereas this is not the case in the construction of subspace codes. In general, given $\fbu(\al_i)$ one can not compute $\fbutwo (\al_i)$. This enforces the modification of KK construction which will be elaborated through this section.

\subsection{A basic generalization of KK codes}
\label{sec:three-A}

Based on the foregoing discussion the first attempt for a simple generalization of KK construction which enables a list-$2$ decoding is explained as follows. The message vector $\bu=(u_0,u_1,\dots,u_{k-1})$ consists of $k$ information symbols over $\F$. Let $\fbu(X)=\sum^{k-1}_{i=0} u_iX^{\left[i\right]}$ be the corresponding message polynomial. Then the corresponding codeword $V$ is the vector space spanned by the set $\left\{(\alpha_1,\fbu(\alpha_1), \fbutwo(\al_1)),\dots,(\alpha_n,\fbu(\alpha_n),\fbutwo(\al_n))\right\}$. Since $\alpha_i$'s are linearly independent, $V$ has dimension $n$. $V$ is transmitted through the operator channel and another vector space $U$ of dimension $r$ is received at the receiver. Let $(x_i,y_i,z_i),i=1,\dots,r$, be a basis for $U$. At the decoder, we construct a nonzero trivariate linearized polynomial $Q(X,Y,Z)$ of the form
\be{Q3-form}
Q(X,Y,Z)=Q_0(X)+Q_1(Y)+Q_2(Z)
\ee
where $Q_i$'s are linearized polynomials over $\F$ subject to certain degree constraints specified later, such that $Q(x_i,y_i,z_i)=0$ for $i=1,\dots,r$. Since $Q$ is linearized, it is zero over the whole vector space $U$, in particular over the intersection of $V$ and $U$. Therefore, assuming that not too many errors and erasures happen the polynomial 
$$
Q(X,\fbu(X),\fbutwo(X))=Q_0(X)+Q_1 \otimes \fbu(X)+Q_2 \otimes \fbutwo(X)
$$
is guaranteed to have a certain number of linearly independent roots which is more than its $q$-degree. Thus it is identically zero and the next step is to recover the message polynomial from it. The problem is how many possible solutions for $\fbu(X)$ we could have and how to extract them. Unfortunately there might be more than two solutions for $\fbu(X)$. In general, an equation over a non-commutative ring may have more zeros than its degree. We illustrate this for the ring of linearized polynomials in the following example. Consider the following equation:
$$
\ftwo(X)-X^{q^2}=0
$$
This can be regarded as an equation of degree $2$ over the ring of linearized polynomials. Then $f(X)=uX^q$ is a solution for this equation for any $u$ which satisfies $u^{q+1}=1$. If $m$ is even, then $q+1$ divides $q^m-1$. Therefore there are $q+1$ distinct possible values for $u$ each gives a distinct solution for $f(X)$. 

\subsection{Solving the problem of having more than two roots}
\label{sec:three-B}

As discussed in the forgoing subsection, an equation over the ring of linearized polynomials may have more zeros than its root.  This is a a consequence of the fact that the ring of linearized polynomials is not commutative. The idea to solve this problem is to restrict the set of message polynomials to a commutative subring of this ring. \Lref{lemma1} shows that linearized polynomials over the base field $\Fq$, $\L_q[X]$, form a commutative subring of $\L_{q^m}[X]$. \Tref{thm1} proves that an equation of degree $L$ over the ring of linearized polynomials has at most $L$ roots in $\L_{q}[X]$, as expected. This suggests the following solution for the problem of having more than two roots. We only consider message polynomials that are over $\Fq$ rather than $\Fqm$ i.e. we assume that the message is a vector $\bu = (u_0,u_1,\dots,u_{k-1})$ of length $k$ over $\Fq$.

\begin{lemma}
\label{lemma1}
Let $f(X)$ and $g(X)$ be linearized polynomials over $\Fq$. Then they commute i.e.
$$
f(X)\otimes g(X)=g(X)\otimes f(X)
$$
\end{lemma}

\begin{proof}
Let $f(X)=\sum^{}_{i\geq 0} f_iX^{\left[i\right]}$ and $g(X)=\sum^{}_{j\geq 0} g_jX^{\left[i\right]}$. Then by \eq{eq20},
\\$f(X)\otimes g(X)=\sum^{}_{k\geq 0} c_kX^{\left[k\right]}$, where $c_k=\sum^{k}_{i=0}f_ig_{k-i}^{\left[i\right]}$ and
\\\\$g(X)\otimes f(X)=\sum^{}_{k\geq 0} c'_kX^{\left[k\right]}$, where $c'_k=\sum^{k}_{i=0}f_i^{\left[k-i\right]}g_{k-i}$.
\\Since $f_i,g_j\in \Fq$, $f_i^{\left[k-i\right]}=f_i^{q^{k-i}}=f_i$ and $g_{k-i}^{\left[i\right]}=g_{k-i}^{q^i}=g_{k-i}$, for any $i$ and $k$. It implies that for any $k$,
$$
c_k=\sum^{k}_{i=0}f_ig_{k-i}=c'_k
$$
Therefore, $f(X)\otimes g(X)=g(X)\otimes f(X)$. 
\end{proof}

\begin{theorem}
\label{thm1}
Let $Q_i$, $i=0,1,\dots,L$, be linearized polynomials over $\F$, where at least one of them is non-zero. Then the equation 
\be{eq11}
\sum^{L}_{i=0}Q_i \otimes f^{\otimes i}(X)=0
\ee
has at most $L$ solutions for $f(x)$ in $\L_q[X]$.
\end{theorem}

\begin{proof}
We do induction on $L$ for $L \geq 0$. For $L=0$, $Q_0$ has to be non-zero. Thus there is no solution for \eq{eq11}. Now, suppose that it is true for $L-1$ and we want to prove it for $L$. If \eq{eq11} does not have any solution for $f(X)$, then we are done. Otherwise, let $f_0(X)$ be a solution for \eq{eq11} i.e.
\be{eq12}
\sum^{L}_{i=0}Q_i \otimes f_0^{\otimes i}(X)=0
\ee
We show that there are at most $L-1$ other solutions excluding $f_0$. 
Subtracting \eq{eq12} from \eq{eq11} we get
\be{eq13}
\sum^{L}_{i=1}Q_i\otimes(f^{\otimes i}-f_0^{\otimes i})=0
\ee
Since $f$ and $f_0$ are both over $\Fq$, by \Lref{lemma1} they commute. As a result,
$$
f^{\otimes i}-f_0^{\otimes i}=(\sum^{i-1}_{j=0}f_0^{\otimes (i-j-1)}\otimes f^{\otimes j})\otimes (f-f_0)
$$
for any $i \geq 1$. Plugging in this into \eq{eq13} we get
\begin{equation*}
\begin{split}
&\sum^{L}_{i=1}Q_i\otimes \bigl( \sum^{i-1}_{j=0} (f_0^{\otimes (i-j-1)}\otimes f^{\otimes j})\otimes (f-f_0) \bigr)=0 \Rightarrow\\
&\bigl(\sum^{L}_{i=1}Q_i\otimes \sum^{i-1}_{j=0}f_0^{\otimes (i-j-1)}\otimes f^{\otimes j} \bigr)\otimes(f-f_0)=0
\end{split}
\end{equation*}
Since $f-f_0\neq 0$, we can divide by both sides by $f-f_0$ to get
\begin{equation*}
\begin{split}
&\sum^{L}_{i=1}\bigl(Q_i\otimes \sum^{i-1}_{j=0} f_0^{\otimes (i-j-1)}\otimes f^{\otimes j}\bigr)=0 \Rightarrow\\
&\sum^{L-1}_{j=0}\bigl(\sum^{L}_{i=j+1}Q_i \otimes f_0^{\otimes (i-j-1)}\bigr)\otimes f^{\otimes j}=0
\end{split}
\end{equation*}
which has at most $L-1$ solutions for $f(X)$ by induction hypothesis. This completes the proof.
\end{proof}

\subsection{Solving the rate reduction problem}
\label{sec:three-C}

As discussed in the forgoing subsection, we suppose that the message vector $\bu=(u_0,\dots,u_{k-1})$ consists of $k$ information symbols over $\Fq$ rather than $\Fqm$. This results in a reduction in rate by a factor of $m$. In this subsection, we propose a solution for the rate reduction problem. Indeed, we take advantage of the fact that the message polynomial is over the base field $\Fq$ in order to compensate the rate reduction at the decoder. 

Recall from \cite[Ch. 4.9]{MS} that any finite extension $\Fqm$ of $\Fq$ as a vector space over $\Fq$ has a basis of the form $\alpha,\alpha^q,\dots,\alpha^{q^{m-1}}$ for some primitive element $\alpha \in \Fqm$. This is called a normal basis for $\Fqm$ over $\Fq$. Suppose that $f(X)$ is a linearized polynomial over $\Fq$. Then for any $j$, $f(\alpha^{q^j})=f(\alpha)^{q^j}$. This implies that given $f(\alpha)$ one can determine $f(\alpha^{q^j})$, for $j=1,2,\dots,m-1$. Therefore, $f(\al^q), f(\al^{q^2}), \dots, f(\al^{q^{m-1}})$ do not need to be transmitted. The idea is to manufacture them at the receiver while only $f(\al)$ is transmitted. We elaborate this idea in the next subsection by specifying encoding and list-decoding algorithm. 

\subsection{Encoding and List-decoding}
\label{sec:three-D}

We fix the following parameters of the code: finite field $\Fq$, an extension $\F=\Fqm$, number of information symbols $k$ and $\al \in \F$ which generates a normal basis for $\F$. We require that $k \leq m$. The ambient space $W$ is a $2m+1$-dimensional vector space over $\Fq$.  
\\\textbf{Encoding Algorithm A:}
 \\Formally, the encoder is a function $\cE\! : \Fq^k \! \to \cG(W,n)$. The input to the encoder is a message $\bu =(u_0,u_1,\dots,u_{k-1}) \in \Fq^k$. The corresponding message polynomial is $\fbu (X) = \sum^{k-1}_{i=0} u_i X^{\left[i\right]}$. The output of the encoder is the one dimensional subspace $V$ as follows:
 $$
 V=\left\langle \bigl(\al,\fbu(\al),\fbutwo(\al)\bigr)\right\rangle
 $$ 
\\\textbf{Remark.\,} We write each element of the ambient space $W$ as a vector with $3$ coordinates such as $(x,y,z)$, where $x \in \left\langle \al\right\rangle$ and $y,z \in \Fqm$. 
\hfill\raisebox{-0.5ex}{$\Box$}\vspace{1.0ex}

Suppose that $V$ is transmitted through the operator channel and another subspace $U$ of $W$ of dimension $1+t$ is received, where $t$ is the dimension of error. We assume that no erasure happens as  only one erasure may destroy all the information. The decoder first check if the following condition on $t$ is satisfied:
\be{error-bound2}
t < 2 - \frac{3(k-1)}{m}
\ee
If not, then the decoder declares decoding failure. Otherwise, the list-decoding algorithm A is performed. 
\\\textbf{List-decoding Algorithm A:}
\\ The input to the decoder is the received vector space $U$. It outputs a list of size at most $2$ of vectors in $\Fq^k$ in three steps:
\begin{enumerate}
\item \textit{Computing the interpolation points:} The decoder first finds a basis $(x_i,y_i,z_i), i=1,2,\dots,t+1$ for $U$. Then the set of interpolation points $\cP$ is as follows:
$$
\cP = \left\{\bigl(x_i^{q^j},y_i^{q^j},z_i^{q^j}\bigr) : 1 \leq i \leq t+1, 0 \leq j \leq m-1 \right\}
$$

\item \textit{Interpolation:} Construct a nonzero trivariate linearized polynomial $Q(X,Y,Z)$ of the form in \eq{Q3-form}, where $Q_0$, $Q_1$ and $Q_2$ are linearized polynomials over $\F$ and $Q_0$ has $q$-degree at most $m-1$, $Q_1$ has $q$-degree at most $m-k$ and $Q_2$ has $q$-degree at most $m-2k+1$, subject to the constraint that
\be{interpolation2}
Q(x,y,z)=0 
\ee
for any $(x,y,z) \in \cP$.
\item \textit{Factorization:} Find all the roots $f(X) \in \L_q[X]$, with degree at most $k-1$ of the equation: 
$$
Q\bigl(X,f(X),\ftwo(X)\bigr) = 0
$$
using LRR algorithm, discussed in the appendix. The decoder outputs coefficients  of each root $f(X)$ as a vector of length $k$.
\end{enumerate}

We explain how the several steps of this list-decoding approach can be done in its most general case in Section \ref{sec:five}. 
\begin{theorem}
\label{thm-decoding-radius2}
List-decoding algorithm A produces a list of size at most $2$ which includes the transmitted message $\bu$ if the number of errors $t$ is less than $ 2- \frac{3(k-1)}{m}$. 
\end{theorem}

We omit the proof of this theorem as it is a special case of \Tref{thm-decoding-radius}. Indeed list-decoding algorithm with general list size will be discussed in the next section which includes list-2 decoding as a special case. Notice that \Tref{thm-decoding-radius2} shows \eq{error-bound2} is also sufficient for successful decoding. This provides the error decoding radius of list-decoding algorithm A.

\subsection{Code parameters}
\label{sec:three-E}

The ambient space $W$, in the construction proposed by encoding algorithm A, has dimension $2m+1$. Each codeword is a one dimensional subspace of $W$. Thus $n=1$ and the packet rate $\Rs$ of the code is given as follows:
$$
\Rs=\frac{\text{log}_{q^m}(\left|\cC\right|)}{n}=\frac{k}{m}
$$
The $q$-degree of $Q_2$ has to be non-negative which enforces the condition $2k \leq m$. It implies that the packet rate $\Rs$ has to be less than $\frac{1}{2}$. The decoding radius normalized by dimension $n=1$, is given by $\tau = 2-\frac{3(k-1)}{m}$ as a result of \Tref{thm-decoding-radius2}. This is roughly equal to $2-3\Rs$. Since $\tau$ is integer in this case, $\tau=1$  for $\Rs < \frac{1}{3}$ and otherwise, $\tau=0$.

%=======================================================================%
%                                                                       %
%    One dimensional codes with general list size                       %              
%                                                                       %
%=======================================================================%

\section{One Dimensional Codes with General List Size }
\vspace{.25ex}
\label{sec:four}

In this section, we generalize the encoding and list-2 decoding algorithm proposed in the foregoing section to general list size yet the construction is one dimensional. To this end, we transmit all powers of $\fbu(X)$ up to $\fl(X)$, where $\fbu$ is the message polynomial, in order to do list-$L$ decoding at the receiver. 

\subsection{Encoding and Decoding}
The following parameters of the code are fixed: finite field $\Fq$, an extension $\F=\Fqm$, number of information symbols $k$, list size $L$ and $\al \in \F$ which generates a normal basis for $\F$. The required condition is that $k \leq m$. The ambient space $W$ is an $Lm+1$-dimensional vector space over $\Fq$.  
\\\textbf{Encoding Algorithm B:}
 \\Formally, the encoder is a function $\cE\! : \Fq^k \! \to \cG(W,n)$. It accepts as input a message vector $\bu =(u_0,u_1,\dots,u_{k-1}) \in \Fq^k$. The message polynomial is constructed as $\fbu (X) = \sum^{k-1}_{i=0} u_i X^{\left[i\right]}$. Then the encoder outputs the following one dimensional subspace $V$:
 $$
 V=\left\langle \bigl(\al,\fbu(\al),\fbutwo(\al),\dots, \fl(\al)\bigr)\right\rangle
 $$ 
 \\\textbf{Definition\,3.} The code $\cC_q(k,1,m,L)$ is the collection of all possible codewords $V$ generated by this encoding algorithm. The second parameter stands for the dimension of the code which is equal to $1$ for this code. 
\hfill\raisebox{-0.5ex}{$\Box$}\vspace{1.0ex}
\\\textbf{Remark.\,} Each element of the ambient space $W$ is indicated as a vector with $L+1$ coordinates such as $(x,y_1,y_2,\dots,y_L)$, where $x \in \left\langle \al\right\rangle$ and all other coordinates are elements of $\Fqm$. 
\hfill\raisebox{-0.5ex}{$\Box$}\vspace{1.0ex}

Suppose that $V$ is transmitted through the operator channel and another subspace $U$ of $W$ of dimension $1+t$ is received, where $t$ is the dimension of error. We assume that no erasure happens as  only one erasure may destroy all the information. The decoder first check if the following condition on $t$ is satisfied:
\be{error-bound-L}
t  <  L-\frac{L(L+1)}{2}\frac{(k-1)}{m}
\ee
If not, then the decoder declares decoding failure. Otherwise, the decoder performs the list-decoding algorithm B. 
\\\textbf{List-decoding Algorithm B:}
\\ The decoder accepts as input the received vector space $U$. The output is a list of size at most $L$ of vectors in $\Fq^k$ after the following three steps:
\begin{enumerate}
\item \textit{Computing the interpolation points:} Find a basis $(x_i,y_{i,1},\dots,y_{i,L}), i=1,2,\dots,t+1$ for $U$. Then the following set is the set of interpolation points $\cP$:
$$
 \left\{\bigl(x_i^{q^j},y_{i,1}^{q^j},\dots,y_{i,L}^{q^j}\bigr) : 1 \leq i \leq t+1, 0 \leq j \leq m-1 \right\}
$$

\item \textit{Interpolation:} Construct a nonzero multivariate linearized polynomial $Q(X,Y_1,\dots,Y_L)$ of the form 
$$
Q_0(X)+Q_1(Y_1)+\dots+Q_L(Y_L)
$$
with each $Q_i$ having $q$-degree at most $m-(k-1)i-1$, for $i=0,1,\dots,L$, subject to the constraint that
\be{interpolation-L}
Q(x,y_1,\dots,y_L)=0 
\ee
for any $(x,y_1,\dots,y_L) \in \cP$.
\item \textit{Factorization:} Find all the roots $f(X) \in \L_q[X]$, with degree at most $k-1$ of the equation: 
\be{factorization-L}
Q\bigl(X,f(X),\dots,f^{\otimes L(X)}\bigr) = 0
\ee
using LRR algorithm. The decoder outputs coefficients  of each root $f(X)$ as a vector of length $k$.
\end{enumerate}

\subsection{Correctness of the algorithm}

\begin{lemma}
\label{lemma10}
There is a non-zero solution for multivariate linearized polynomial $Q$ which satisfies \eq{interpolation-L} provided that
$$
t < L-\frac{L(L+1)}{2}\frac{(k-1)}{m}
$$
\end{lemma}
\begin{proof}
The set of interpolation points $\cP$ contains $m(1+t)$ points. Therefore, \eq{interpolation-L} defines a homogeneous linear system of $m(1+t)$ equations. The number of unknown coefficients is equal to
$$
\sum^{L}_{i=0} m-(k-1)i=(L+1)m-(k-1)\frac{L(L+1)}{2}
$$
It is known that if the number of variables in a homogeneous linear system of equation is strictly smaller than the number of equation, then there is a non-trivial solution. Furthermore, this is necessary in order to guarantee a non-trivial solution i.e.  
$$
m(1+t)< (L+1)m-(k-1)\frac{L(L+1)}{2}
$$
guarantees a non-zero solution for $Q$. This is equivalent to 
$$
t < L-\frac{L(L+1)}{2}\frac{(k-1)}{m}
$$
\end{proof}
\begin{corollary}
\label{cor1}
The condition \eq{error-bound-L} is necessary to guarantee existence of interpolation polynomial $Q$ in list-decoding algorithm B.
\end{corollary}

Let $\fbu(X)$ be the message polynomial and $Q(X,Y_1,\dots,Y_L)$ be the interpolation polynomial provided by list-decoding algorithm B. Then let $E(X)$ be the following linearized polynomial:
$$
E(X)=Q\bigl(X,\fbu(X),\dots, \fl(X)\bigr)=\sum^{L}_{i=0}Q_i \otimes \fbu^{\otimes i}(X)
$$
\begin{lemma}
\label{lemma11}
For $j=0,1,\dots,m-1$, $\al^{q^j}$ is a root of $E(X)$.
\end{lemma}
\begin{proof}
Since we assume that no erasure occurs, the transmitted codeword $V$ is contained in the received subspace $U$. In particular, $U$ includes the vector $\bigl(\al,\fbu(\al),\dots,\fl(\al)\bigr)$. Notice that raising to power $q^j$ is a linear operation. Therefore, $\bigl(x^{q^j},y_1^{q^j},\dots,y_L^{q^j}\bigr)$ is a linear combination of some elements of the set of interpolation points $\cP$, for any $(x,y_1,\dots,y_L) \in U$, as $\cP$ contains all the $q^j$-powers of the basis elements of $U$. Furthermore, $Q$ is a linearized polynomial. Therefore,
$$
Q\bigl(x^{q^j},y_1^{q^j},\dots,y_L^{q^j}\bigr)=0
$$
In particular,
\be{alpha-root}
Q \bigl(\al^{q^j},\fbu(\al)^{q^j},\dots,\fl(\al)^{q^j}\bigr) = 0
\ee
Note that for any polynomial $f(X) \in \Fq[X]$, $f(X^{q^j})=f(X)^{q^j}$ for any positive integer $j$. Since all the coefficients of $\fbu^{\otimes i}(X)$ are elements of $\Fq$, \eq{alpha-root} implies that
$$
E(\al^{q^j})=Q \bigl(\al^{q^j},\fbu(\al^{q^j}),\dots,\fl(\al^{q^j})\bigr) = 0
$$
\end{proof}

\begin{corollary}
\label{cor2}
$E(X)$ is the all zero polynomial.
\end{corollary}
\begin{proof}
Since the $q$-degree of $\fbu(X)$ is at most $k-1$, the $q$-degree of $Q_i \otimes \fbu^{\otimes i}(X)$ is at most 
$$
m-(k-1)i-1+(k-1)i=m-1,
$$
for $i=0,1,\dots,L$. This implies that $q$-degree of $E(X)$ is at most $m-1$. On the other hand, $E(X)$ has at least $m$ linearly independent roots $\al,\al^q,\dots,\al^{q^{m-1}}$ by \Lref{lemma11}. Therefore, $E(X)$ must be the all zero polynomial.
\end{proof}

\begin{theorem}
\label{thm-decoding-radius}
List-decoding algorithm B produces a list of size at most $L$ which includes the transmitted message $\bu$ provided that
$$
t < L-\frac{L(L+1)}{2}\frac{(k-1)}{m}
$$
\end{theorem}

\begin{proof}
By \Lref{lemma10} the interpolation polynomial $Q \neq 0$ exists. Then by \Cref{cor2}, $E(X)$ is identically zero which means that the message polynomial $\fbu(X)$ is a solution to \eq{factorization-L}. Also, as $Q$ is nonzero, \eq{factorization-L} has at most $L$ solutions by \Tref{thm1}. Therefore, the list size is at most $L$.
\end{proof}

\subsection{Code parameters:}

In this subsection, we discuss the parameters of the code $\cC_q(k,1,m,L)$. The ambient space $W$ has dimension $Lm+1$. Each codeword is a one dimensional subspace of $W$. Therefore, $n=1$ and the packet rate $\Rs$ of the code is given as follows:
$$
\Rs=\frac{\text{log}_{q^m}(\left|\cC_q(k,1,m,L)\right|)}{n}=\frac{k}{m}
$$
The $q$-degree of $Q_L$, the one with smallest degree among $Q_i$'s, has to be non-negative which leads to the following series of inequalities:
\begin{equation*}
\begin{split}
m-(k-1)L-1&\geq 0 \Rightarrow \\
L &\leq \frac{m-1}{k-1} \approx \frac{1}{\Rs} \Rightarrow \\
\Rs &\leq \frac{1}{L}
\end{split}
\end{equation*}

\Cref{cor1} and \Tref{thm-decoding-radius} together show that the bound on the number of errors in \eq{error-bound-L} is a necessary and sufficient condition in order to guarantee correct list-decoding. Therefore, the error decoding radius of list-decoding algorithm B is equal to  $L-\frac{L(L+1)}{2}\frac{(k-1)}{m}$. Since the dimension of the code is equal to $1$, this is also the normalized decoding radius, denoted by $\tau$. We can approximate $\frac{k-1}{m}$ by $\Rs$ to get the following approximation for the normalized decoding radius $\tau$:
$$
\tau \approx L - \frac{L(L+1)}{2}\Rs
$$

\vspace{1ex}

%=======================================================================%
%                                                                       %
%    Codes with Arbitrary Dimension                                    %          
%                                                                       %
%=======================================================================%
\section{Codes with Arbitrary Dimension}
\vspace{.25ex}
\label{sec:five}

In the foregoing section we proposed a list-decodable construction of subspace codes along with corresponding list-$L$ decoding algorithm for any list size $L$. One weakness of the construction is that codes are one dimensional. Certainly this is not good. One dimensional codes seem somewhat unnatural. Besides, as the normalized decoding radius $\tau$ is integer in this case, we are not able to take advantage of the whole achievable region for $\tau$. In this section, we generalize our construction to any arbitrary dimension. 

In the construction of $\cC_q(k,1,m,L)$, span of $(\al,\fbu(\alpha),\dots$ $, \fl(\al))$ is the codeword corresponding to message polynomial $\fbu$, where $\al$ is generator of a normal basis for $\Fqm$.  Then one simple way to generalize this construction to dimension $2$ is the following. Suppose that $\beta$ is another primitive element of $\Fqm$ which generates a normal basis for $\Fqm$. Then the corresponding codeword is the $\Fq$-linear span of $\bigl(\al,\fbu(\alpha),\dots$ $, \fl(\al)\bigr)$ and $\bigl(\beta,\fbu(\beta),\dots,\fl(\beta)$. When we inject more vectors into the network, we indeed add \emph{redundancy} to the code, and we should get something in return. Adding redundancy means more interpolation points at the receiver which indeed enforces more constraints. In return, we should get more zeros in order to maintain same performance in terms of decoding radius versus rate. However, as we shall see in \Lref{lemma11}, the space root already covers the whole space $\Fqm$. Therefore, this simple generalization does not lead to a good performance. This becomes even worse as dimension increases. The idea is to evaluate the interpolation polynomial in a larger field i.e. an extension of $\Fqm$. We elaborate this idea through this section.  

Fix a finite field $\Fq$ and let $n$ divides $q-1$. Then the equation $x^n-1=0$ has $n$ distinct solutions in $\Fq$. Let $e_1=1,e_2,e_3,\dots,e_n$ be these solutions.  Let $\F=GF(q^{nm})$ and $\gamma$ be a generator of a normal basis for $\F$. Then define
\be{define-alpha}
\alpha_i=\gamma+e_i^{-1}\gamma^{q^m}+e_i^{-2}\gamma^{q^{2m}}+\dots+e_i^{-(n-1)}\gamma^{q^{(n-1)m}}
\ee
for $i=1,2,\dots, n$. 

Next, we discuss the properties of the parameters $\al_i$'s. 
\begin{lemma}
\label{lemma5}
$\al_1 \in \Fqm$ and for $i=2,3,\dots,n$, $\alpha_i^n \in \Fqm$.
\end{lemma}

\begin{proof}
For $i=1,2,\dots,q-1$, $\al_i^{q^m} = e_i^{-1}\al_i$ by the following series of equalities:
\begin{equation*}
\begin{split}
\alpha_i^{q^m}&=\bigl(\sum^{n-1}_{j=0}e_i^j\gamma^{q^{mj}}\bigr)^{q^m}=\sum^{n-1}_{j=0}(e_i^{q^m})^j\gamma^{q^{m(j+1)}}=\sum^{n-1}_{j=0}e_i^j\gamma^{q^{m(j+1)}}\\
&=e_i^{n-1}\gamma^{q^{nm}}+\sum^{n-1}_{j=1}e_i^{j-1}\gamma^{q^{mj}}=e_i^{-1}\gamma+\sum^{n-1}_{j=1}e_i^{j-1}\gamma^{q^{mj}}\\
&=\sum^{n-1}_{j=0}e_i^{j-1}\gamma^{q^{mj}}=e_i^{-1}\alpha_i
\end{split}
\end{equation*}
Then for $i=1$, $\al_1^{q^m} = \al_1$ and therefore, $\al_1 \in \Fqm$. For $i=2,3,\dots,n$,  $(\alpha_i^{n})^{q^m}=e_i^{-n}\alpha_i^{n}=\alpha_i^n$ which implies that $\alpha_i^n\in \Fqm$.
\end{proof}

\begin{lemma}
\label{lemma6}
The set 
$$
Z=\left\{\alpha_i^{q^j}:1 \leq i \leq n, 0 \leq j \leq m-1 \right\}
$$ 
is a basis for $\F$.
\end{lemma}

\begin{proof}
Let $A$ and $\Gamma$ be $1 \times n$ vectors as follows:
\begin{equation*}
\begin{split}
A&=(\alpha_1,\alpha_2,\dots,\alpha_n)\\
\Gamma&=\bigl(\gamma,\gamma^{q^m},\dots, \gamma^{q^{(n-1)m}}\bigr)
\end{split}
\end{equation*}
Also, let $E$ be the $n \times n$ matrix whose $(i,j)$ entry is $e_i^{j-1}$. Then by definition, $A=\Gamma E^t$. Note that $E$ is a Vandermonde matrix whose determinant is non-zero, since $e_i$'s are distinct. Thus, we can write 
$$
\Gamma=A(E^{-1})^t
$$ 
It implies that for any $j$, $\gamma^{q^{mj}}$ is a linear combination of $\alpha_i$'s. We can raise this to power $q^r$, for any $0 \leq r \leq m-1$, and say that $\gamma^{q^{mj+r}}$ is a linear combination of $\alpha_i^{q^r}$'s. Thus $\gamma^{q^l}$ is a linear combination of elements of the set $Z$, for $0 \leq l \leq nm-1$. Therefore, elements of $Z$ span the whole space $\F$. But $\left|Z\right|=nm$. Hence $Z$ is a basis for $\F$.
\end{proof}

\subsection{Encoding and Decoding}

The following parameters of the construction are fixed: the finite field $\Fq$ and an extension field $\Fqm$, the number of information symbols $k$, the dimension of code $n$ and the list size $L$. We require that $k \leq nm$ and $n \leq q-1$. 

We let $[s]$ denote the set of positive integers less than or equal to $s$, for any positive integer $s$. 
\\\textbf{Encoding Algorithm C:}
 \\A message vector $\bu =(u_0,u_1,\dots,u_{k-1}) \in \Fq^k$ is the input to the encoder. The corresponding message polynomial is $\fbu (X) = \sum^{k-1}_{i=0} u_i X^{\left[i\right]}$. For $i=1,2,\dots,n$, consider $\al_i$ defined in \eq{define-alpha}. The encoder constructs vectors $v_i \in W$ as follows. For $i=1,2,\dots,n$, 
$$
v_i=(\alpha_i,\fbu(\alpha_i),\fbu^{\otimes 2}(\alpha_i),\dots, \fbu^{\otimes L}(\alpha_i))
$$
The encoder then outputs $n$-dimensional vector space $V$ spanned by $v_1,v_2,\dots,v_n$. 

In this construction, the ambient space $W$ is 
\be{ambient-space}
\left\langle \al_1,\al_2,\dots,\al_n \right\rangle \oplus \underbrace{\F_{q^{nm}}\oplus \dots \oplus \F_{q^{nm}}}_{L \text{ times}}
\ee
which has dimension equal to $n+nmL$. 
\\\textbf{Remark.\,} 
Each element in $W$ is represented by a vector with $L+1$ coordinates such as $(x,y_1,y_2,\dots,y_L)$, where $x$ belongs to the vector space spanned by $\al_1,\al_2,\dots,\al_n$ and $y_i \in \F_{q^{nm}}$, for $i=1,2,\dots,L$.
\hfill\raisebox{-0.5ex}{$\Box$}\vspace{1.0ex}
 \\\textbf{Definition\,4.} The code $\cC_q(k,n,m,L)$ is the collection of all possible codewords $V$ generated by encoding algorithm C.
 \hfill\raisebox{-0.5ex}{$\Box$}\vspace{1.0ex}
 
Suppose that a codeword $V \in \cC_q(k,n,m,L)$ is transmitted through the operator channel and the decoder receives a vector space $U \in \cG(W)$ with dimension $d$. Then we define the parameter $\omega$ as follows:
\be{define-omega}
\omega = \left\lceil \frac{md+1}{L+1} + \frac{1}{2}L(k-1) \right\rceil
\ee
This will guarantee existence of the interpolation polynomial $Q$ in the list-decoding algorithm C. The algorithm is as follows:
 \\\textbf{List-decoding Algorithm C:}
\begin{enumerate}
\item \textit{Computing the interpolation points:} Find a basis for $U$ as follows: 
$$
\left\{(x_i,y_{i,1},y_{i,2},\dots,y_{i,L}): i=1,2,\dots,d\right\}
$$
Then for $h=0,1,2,\dots,m-1$, the set $\cP_h$ is defined as follows:
$$
\cP_h = \left\{(x_i^{q^h},y_{i,1}^{q^h},\dots,y_{i,L}^{q^h}): i \in [d] \right\}
$$
The set of interpolation points $\cP$ is equal to: 
$$
\cP = \cP_1 \cup \cP_2 \cup \dots \cup \cP_h
$$

\item \textit{Interpolation:} Construct a nonzero multivariate linearized polynomial $Q(X,Y_1,Y_2,\dots$ $,Y_L)$ of the form
$$
Q_0(X)+Q_1(Y_1)+Q_2(Y_2)+\dots+Q_L(Y_L)
$$
with each $Q_i$ having $q$-degree at most $\omega-(k-1)i-1$, for $i=0,1,\dots,L$ subject to the constraint that
\be{interpolation-general}
Q(x,y_1,y_2,\dots,y_L)\, =\, 0
\ee
for any $(x,y_1,y_2,\dots,y_L) \in \cP$. 

\item \textit{Factorization:} Find all the roots $f(X) \in \L_q[X]$, with degree at most $k-1$ of the equation: 
\be{factorization-general}
Q\bigl(X,f(X),\dots,f^{\otimes L}(X)\bigr) = 0
\ee
using LRR algorithm. The decoder outputs coefficients  of each root $f(X)$ as a vector of length $k$.
\end{enumerate}

The first step of the list-decoding algorithm C is done by elementary linear algebraic operations. The interpolation step is indeed solving a system of linear of equations. There are several ways for doing that. The most straightforward way is the Gaussian elimination method. However, this method does not take advantage of the certain structure of this system of equations and therefore, it is not efficient. An efficient polynomial-time interpolation algorithm in the ring of linearized polynomials is presented in \cite{XYS} which is basically analogous to Koetter interpolation algorithm in the ring of polynomials. The factorization step can be performed using linearized Roth-Ruckenstein algorithm, called LRR algorithm, which will be explained in details in the appendix. We have modified the Roth-Ruckenstein algorithm \cite{RR} in order to solve equations over the ring of linearized polynomials. For instance, an equation of degree $L$ over $\L_{q^m}[X]$ has the following form: 
$$
Q_0(X)+Q_1(X) \otimes f(X)+ \dots +Q_L(X) \otimes f^{\otimes L}(X) = 0
$$
where $Q_i$'s are linearized polynomials over $\Fqm$. LRR algorithm finds all the roots of this equation in efficient polynomial time.  

\subsection{Correctness of the algorithm}

\begin{lemma}
\label{lemma14}
Existence of a non-zero solution for interpolation polynomial $Q$ that satisfies \eq{interpolation-general} is guaranteed by the choice of $\omega$ in \eq{define-omega}. 
\end{lemma}

\begin{proof}
\eq{interpolation-general} defines a homogeneous system of at most $md$ equations. The number of unknown coefficients is as follows:
$$
\sum^{L}_{i=0} \omega -(k-1)i=(L+1)\omega-(k-1)\frac{L(L+1)}{2}
$$
A non-zero solution for this homogeneous linear system of linear equations is guaranteed if and only if the number of equations is strictly less than the number of variables. i.e.
\begin{equation*}
\begin{split}
md &<  (L+1)\omega-(k-1)\frac{L(L+1)}{2} \Leftrightarrow \\
\omega &\geq \frac{md+1}{L+1} + \frac{1}{2}L(k-1) 
\end{split}
\end{equation*}
This is guaranteed by the choice of $\omega$ in \eq{define-omega}.
\end{proof}

\begin{lemma}
\label{lemma15}
The linear spans of the sets $\cP_h$, defined in the first step of list-decoding algorithm C, are disjoint for $h=0,1,\dots,m-1$. 
\end{lemma}

\begin{proof}
For any $i \in [d]$, $x_i$ is an element of the span of $\al_1,\al_2,\dots,\al_n$. Since raising to power $q^h$ is a linear operation, $x_i^{q^h}$ is an element of 
$$
\left\langle \al_1^{q^h},\al_2^{q^h},\dots,\al_n^{q^h} \right\rangle
$$
By \Lref{lemma6}, these are disjoint vector spaces as $h$ varies between $0$ and $m-1$. Therefore, linear spans of $\cP_h$'s are also disjoint as $h$ varies between $0$ and $h-1$.
\end{proof}

We form the following linearized polynomial $E(X)$ wherein $\fbu(X)$ is the message polynomial and $Q(X,Y_1,\dots,Y_L)$ is the interpolation polynomial provided by list-decoding algorithm C. 
$$
E(X)=Q\bigl(X,\fbu(X),\dots, \fl(X)\bigr)=\sum^{L}_{i=0}Q_i \otimes \fbu^{\otimes i}(X)
$$

Suppose that the number of errors in the received subspace $U$ is $t$ and the number of erasures is equal to $\rho$. Thus $d=n-\rho+t$. 

\begin{lemma}
\label{lemma18}
The linearized polynomial $E(X)$ has at least $(n-\rho)m$ linearly independent roots.
\end{lemma}
\begin{proof}
Let $U'$ be the intersection of the transmitted codeword $V$ and the received subspace $U$. Then $U'$ is a subspace of the received vector space $U$ with dimension $n-\rho$. For any $(x,y_1,\dots,y_L) \in U'$ and $h=0,1,\dots,m-1$, $\bigl(x^{q^h},y_1^{q^h},\dots,y_L^{q^h}\bigr)$ is an element of the linear span of the set $\cP_h$, because $\cP_h$ contains all the $q^h$-powers of the basis elements of $U$ and also raising to power $q^h$ is a linear operation. Furthermore, $Q$ is a linearized polynomial. Hence,
\be{18-1}
Q\bigl(x^{q^h},y_1^{q^h},\dots,y_L^{q^h}\bigr)=0
\ee
On the other hand, $(x,y_1,\dots,y_L)$ is also an element of the transmitted codeword $V$. Therefore, 
$$
(x,y_1,\dots,y_L) = \bigl(\beta,\fbu(\beta),\dots, \fl(\beta)\bigr)
$$
for some $\beta$ in the linear span of $\al_1,\al_2,\dots,\al_n$.  Since coefficients of $\fbu(X)$ are elements of $\Fq$, for any integer $h$
\be{18-2}
\bigl(x^{q^h},y_1^{q^h},\dots,y_L^{q^h}\bigr) = \bigl(\beta^{q^h},\fbu(\beta^{q^h}),\dots, \fl(\beta^{q^h})\bigr)
\ee
Notice that linear spans of the sets $\cP_h$ are disjoint by \Lref{lemma15}. This together with \eq{18-1} and \eq{18-2} imply that there are at least $(n-\rho)m$ linearly independent roots for $E(X)$. 
\end{proof}

\begin{corollary}
\label{cor4}
 If $\omega \leq (n-\rho)m$, then the linearized polynomial $E(X)$ is identically zero.
\end{corollary}
\begin{proof}
The $q$-degree of $\fbu(X)$ is at most $k-1$. Therefore, the $q$-degree of $Q_i \otimes \fbu^{\otimes i}(X)$ is at most 
$$
\omega-(k-1)i-1+(k-1)i=\omega-1,
$$
for $i=0,1,\dots,L$. Thus the $q$-degree of $E(X)$ is at most $\omega-1$. On the other hand, $E(X)$ has at least $(n-\rho)m$ linearly independent roots by \Lref{lemma18}. Therefore, $E(X)$ must be the all zero polynomial.
\end{proof}

\begin{theorem}
\label{thm-general}
The output of list-decoding algorithm C is a list of size at most $L$ which includes the transmitted message $\bu$ provided that
\be{error-bound-general}
L\rho+t \leq nL-\frac{L(L+1)}{2}\frac{(k-1)}{m} -\frac{1}{m}
\ee
\end{theorem}
\begin{proof}
The existence of non-zero Interpolation polynomial $Q$ that satisfy \eq{interpolation-general} is guaranteed by \Lref{lemma14}. Then by \Cref{cor4}, $E(X)$ is the all zero polynomial provided that
\be{thm-general-1}
\frac{md+1}{L+1} + \frac{1}{2}L(k-1) \leq (n-\rho)m
\ee
where we have used the expression for $\omega$ from \eq{define-omega}. We plug in $d=n-\rho+t$ into \eq{thm-general-1}. Then observe that \eq{thm-general-1} is indeed equivalent to
$$
L\rho+t \leq nL-\frac{L(L+1)}{2}\frac{(k-1)}{m} -\frac{1}{m}
$$
Thus this condition on the number of errors and erasures implies that $E(X)$ is identically zero. Therefore, the message polynomial $\fbu(X)$ is a solution to \eq{factorization-general}. Also, since $Q$ is nonzero, \eq{factorization-general} has at most $L$ solutions by \Tref{thm1}. Therefore, the list size is at most $L$.
\end{proof}

\subsection{Code parameters}
The ambient space $W$ has dimension $n+nmL$ in construction of the code $\cC_q(k,n,m,L)$. Each codeword is an $n$-dimensional subspace of $W$. Then the packet rate $\Rs$ of the code is given as follows:
$$
\Rs=\frac{\text{log}_{q^m}(\left|\cC_q(k,n,m,L)\right|)}{n}=\frac{k}{nm}
$$
The $q$-degree of each $Q_i$ must be non-negative. Notice that $Q_L$ has the smallest degree among $Q_i$'s. This leads to to the following series of inequalities:
\begin{equation*}
\begin{split}
nm-(k-1)L-1&\geq 0 \Rightarrow \\
L &\leq \frac{nm-1}{k-1} \approx \frac{1}{R} \Rightarrow \\
\Rs &\leq \frac{1}{L}
\end{split}
\end{equation*}

\Tref{thm-general} provides the bound on the number of errors and erasures which guarantees correct list-decoding. This implies that the error decoding radius of list-decoding algorithm C is equal to  
$$
nL-\frac{L(L+1)}{2}\frac{(k-1)}{m} - \frac{1}{m}
$$ 
where the number of erasures is weighted by $L$. Normalizing this by dimension $n$ we get the normalized decoding radius $\tau$:
\begin{equation*}
\begin{split}
\tau &= L-\frac{L(L+1)}{2}\frac{(k-1)}{nm} - \frac{1}{nm}\\
 & \approx L - \frac{1}{2} L(L+1) \Rs
\end{split}
\end{equation*}
which is similar to one dimensional case. It means that increasing the dimension does not affect the decoding radius. The normalized decoding radius of KK construction, given in \eq{tau-KK}, is equal to $1-\Rs$.  Let's call the normalized decoding radius of list-decoding algorithm C with list size $L$ to be $\tau_L$. Notice that $\tau_{\text{KK}}$ is indeed equal to $\tau_1$.  We compare  $\tau_{\text{KK}}$ with $\tau_L$ for various amounts of $L$. This is plotted in Figure\,1.

\includegraphics[width=.47\textwidth]{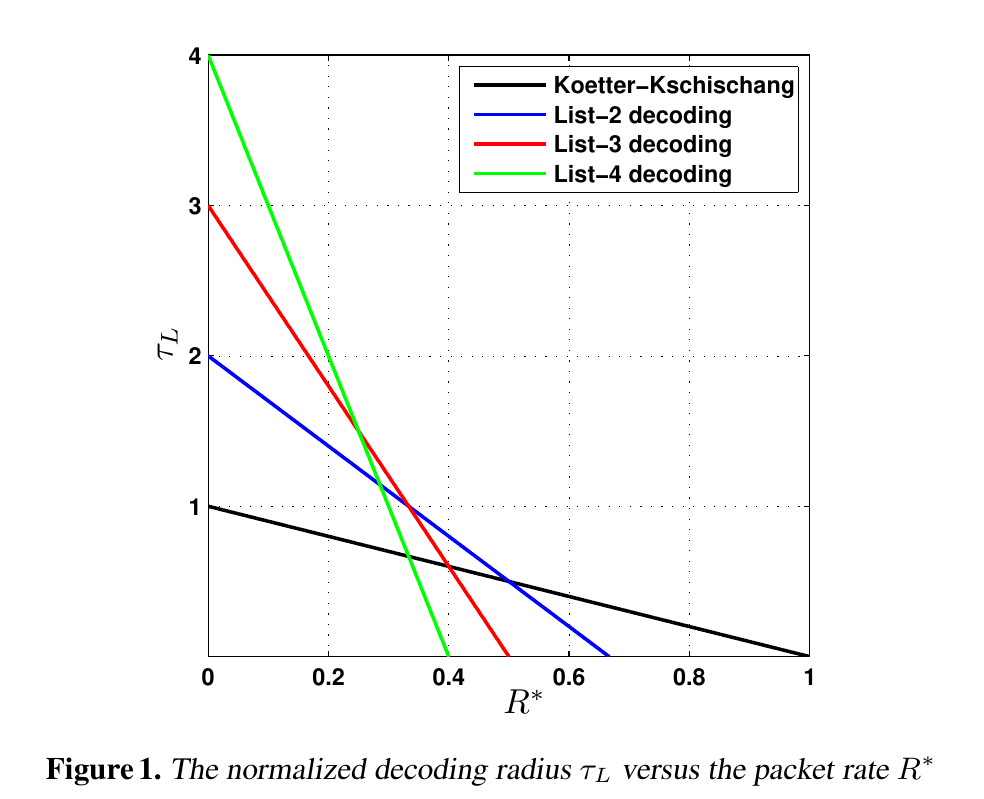}

\vspace{1ex}

%=======================================================================%
%                                                                       %
%    5.Discussions and Conclusion                                        %
%                                                                       %
%=======================================================================%
\section{Discussion and Conclusions}
\vspace{.25ex}
\label{discussion}
In this paper we have considered the problem of list-decoding of subspace codes proposed for error correction in random linear network coding. To this end, we modified and generalized the original Koetter-Kschischang construction in various ways. In fact, we constructed a new subspace code and proposed a list-decoding algorithm that enables error-correction beyond the unique decoding bound. Interestingly, for a fixed code dimension, we can actually correct any number of errors provided that the list size is sufficiently large and the rate is small enough. In this case, the list size is indeed proportional to the number of errors. 

However, we are able to achieve a better decoding radius than KK construction only at low rates. Then one possible direction for future work is to extend this work in order to enable list-decoding at higher rates as well. We may take advantage of similarities between this work and Sudan list-decoding algorithm of RS codes. When Sudan list-decoding algorithm of Reed-Solomon codes was introduced, there was a similar problem. Later Guruswami and Sudan proposed a new method in that they enforced multiple roots for the interpolation polynomial which resulted in a significant improvement upon Sudan's first result. Therefore, it is natural to look for an analogous technique in the ring of linearized polynomials. However, there is no clear notion of multiple roots for linearized polynomials in the literature. In fact, one has to introduce multiplicity in the ring of linearized polynomials in such a way that list-decoding at higher rates is enabled. 

As mentioned, in order to do list-decoding we modify and generalize the KK construction in many ways. Then the natural question is the following: is there a way to list-decode the KK construction without any modification at the transmitter side? This suggests another path for future work. 

\appendix

We provide the linearized Roth-Ruckenstein algorithm (LRR algorithm) which is used in the  factorization step of all of our list-decoding algorithms. LRR algorithm basically solves equations over the ring of linearized polynomials. Consider a polynomial $Q(X,Y)$, where $Y$ is a variable in the ring $\L_q[X]$, of the form
\be{Q-form-LRR}
Q(X,Y) = Q_0(X) + Q_1(X) \otimes Y + \dots + Q_L(X) \otimes Y ^ {\otimes L}
\ee
where $Q_i$'s are linearized polynomials over a finite extension of $\Fq$. LRR algorithm finds all the roots $Y \in \L_q[X]$ with $q$-degree at most $k-1$, for some $k \in \N$, for which $Q(X,Y)$ is identically zero. 

We say that the polynomial $Q(X,Y)$ is divisible by $X^{q^s}$ if all the $Q_i$'s, for $i=1,2,\dots,L$, are divisible by $X^{q^s}$. In this case, for each $i$, there is a linearized polynomial $Q'_i$ such that $Q'_i(X)^{q^s} = Q_i(X)$. Then we define 
$$
Q_{\shortdownarrow s}(X,Y) =  Q'_0(X) + Q'_1(X) \otimes Y + \dots + Q'_L(X) \otimes Y ^ {\otimes L}
$$
\\\textbf{Linearized Roth-Ruckenstein (LRR) algorithm}
\\LRR $(Q(X,Y), k \in \mathbb{N}, \lambda \in \mathbb{N} \cup \left\{0\right\})$
\\Global variables:

set $A \subseteq \L_q[X]$,

polynomial $g(X)=\sum^{k-1}_{i=0} u_i X^{q^i} \in \L_q[X]$.
\\Call procedure initially with $Q(X,Y) \neq 0, k>0,\lambda=0$.
\\if$(\lambda==0)$

$A \leftarrow \emptyset$;
\\$s \leftarrow$ largest integer such that $Q(X,Y)$ is divisible by $X^{q^s}$ 
\\$H(X,\gamma) \leftarrow \frac{1}{X} Q_{\shortdownarrow s}(X,\gamma X)$;
\\$Z \leftarrow$ set of all roots of $H(0,\gamma)$ in $\Fq$;
\\for each $\gamma \in Z$ do \{

$u_{\lambda}\leftarrow \gamma$;

if $(\lambda < k-1)$

\hspace{5mm}LRR$(Q_{\shortdownarrow s}(X,Y^q+\gamma X), k, \lambda +1)$;

else

\hspace{5mm} if $(Q(X,u_{k-1}X)==0)$

\hspace{10mm} $A\leftarrow A \cup \left\{g(X)\right\}$;
\\\}

\begin{lemma}
\label{lemma31}
Let $Q(X,Y)$ be as defined in \eq{Q-form-LRR}. Let 
$$
f(X) = f_0 X+f_1 X^q + \dots + f_{k-1} X^{q^{k-1}}
$$
and
$$
H(X,\gamma)=\frac{1}{X}Q(X,\gamma X)
$$
Then coefficient of $X$ in $Q(X,f(X))$ is equal to $H(0,f_0)$.
\end{lemma}
\begin{proof}
Observe that the coefficient of $X$ in $f^{\otimes i}(X)$ is equal to $f_0^i X$. Therefore, coefficient of $X$ in $Q(X,f(X))$ is equal to coefficient of $X$ in 
$$
Q_0(X)+Q_1(f_0 X)+Q_2(f_0^2X)+\dots+Q_{L}(f_0^{L}X)
$$
The later is equal to $XH(X,f_0)$. Note that coefficient of $X$ in $XH(X,f_0)$ is equal to the constant term in $H(X,f_0)$ which is indeed $H(0,f_0)$.
\end{proof}
Notice that the level of recursion can not go beyond $k-1$. In fact, each sequence of recursions along a recursion descent is associated with a unique polynomial 
$$
f(X)=f_0X+f_1X^q+\dots
$$
which stands for the contents of the global polynomial $g(X)$ computed by that sequence. For $i=0,1,\dots,k-1$, let $P_i(X,Y)$, $T_i(X,Y)$ and $H_i(X,\gamma)$ be the values of $Q(X,Y)$, $Q_{\shortdownarrow s}(X,Y)$ and $H(X,\gamma)$, respectively, during recursion level $\lambda=i$. It can be inductively observed that $P_i$ and $T_i$ are nonzero polynomials for $i=0,1,\dots,k-1$. In fact, $P_0=Q$ is assumed to be nonzero. Since $P_i$ is nonzero, $T_i$ is nonzero which implies that $P_{i+1}$ is nonzero. Therefore, the parameter $s$ is always well-defined. 

At each recursion level $i$, $T_i(X,Y)$ is not divisible by $X^q$. Therefore, coefficient of $X$ in $T_i(X,\gamma X)$ is not zero. Then by \Lref{lemma31}, $H(0,\gamma)$ is not the all zero polynomial.
\begin{lemma}
\label{lemma30}
Let $A$ be the set that is computed by the call LRR$(Q,k,0)$. Then every element of $A$ is a root of $Q$.
\end{lemma}
\begin{proof}
Let 
$$
f(X)=u_0 X+u_1 X^q+\dots+u_{k-1}X^{q^{k-1}}
$$
be an element of $A$. For $0 \leq i < k$, define the polynomial $\phi_i(X)$ by
$$
\phi_i (X)=u_i X+ u_{i+1}X^q+\dots+u_{k-1} X^{q^{k-i-1}}
$$
Since $u_i$'s are elements of $\Fq$, $\phi_i = \phi_{i+1}^q+u_iX$. Let $P_i$ and $T_i$ be the values of $Q$ and $T$ during recursion level $\lambda=i$. We do backward induction on $i=k-1,k-2,\dots,0$ to show that $\phi_i$ is a root of $P_i$. The base of induction is $i=k-1$. Note that $\phi_{k-1}=u_{k-1}X$ which is a root of $P_{k-1}$ by the one before the last line of LRR procedure when $\lambda=k-1$. Now, suppose that $\phi_{i+1}$ is a root of $P_{i+1}$. Then we have
\begin{align*}
P_i(X,\phi_i)=T_i(X,\phi_i)^{q^s}&=T_i(X,\phi_{i+1}^q+u_iX)^{q^s}\\
&=P_{i+1}(X,\phi_{i+1})^{q^s}=0
\end{align*}
Therefore, $\phi_i$ is a root of $P_i$ which completes the induction.  In particular, for $i=0$ we see that $f(X) = \phi_0(X)$ is a root of $P_0=Q$. 
\end{proof}
\begin{lemma}
\label{lemma28}
Let 
$$
f(X)=f_0X+f_1X^q+\dots+f_{k-1}X^{q^{k-1}}
$$
be a root of $Q(X,Y)$ in $\L_q[X]$. For $0 \leq i \leq k-1$, define $P_i(X,Y)$ and $T_i(X,Y)$ inductively by $P_0=Q$ and
$$
T_i(X,Y)^{q^{s_i}}=P_i(X,Y)\ \text{and}\ P_{i+1}(X,Y)=T_i(X,Y^q+f_iX),
$$
where $s_i$ is the largest possible integer such that $P_i(X,Y)$ is divisible by $X^{q^{s_i}}$. Also, define
$$
H_i(X,\gamma)=\frac{1}{X}T_i(X,\gamma X)
$$
Then for $0 \leq i \leq k-1$,
\\i)The polynomial $\phi_i$ defined by
$$
\phi_i=f_iX+f_{i+1}X^q+\dots+f_{k-1}X^{k-1-i}
$$
is a root of $P_i(X,Y)$.
\\ii)$H_i(0,f_i)=0$
\end{lemma}
\begin{proof}
We prove part i) by induction on $i$. For $i=0$, $\phi_0=f$ is a root of $P_0=Q$. Now, suppose that $\phi_i$ is a root of $P_i(X,Y)$. Since $\phi_i=\phi_{i+1}^q+f_iX$, $Y=\phi_{i+1}$ is a root of $P_i(X,Y^q+f_iX)$ and, hence, of $T_i(X,Y^q+f_iX)=P_{i+1}(X,Y)$. This completes the induction which proves part i).

Also, note that
$$
P_i(X,\phi_i(X))=T_i(X,\phi_i(X))^{q^{s_i}}=0\ \Rightarrow\ T_i(X,\phi_i(X))=0
$$
By \Lref{lemma31}, the coefficient of $X$ in $T_i(X,\phi_i(X))$ is equal to $H_i(0,f_i)$ which has to be zero. This proves part ii).
\end{proof}
\begin{lemma}
\label{lemma29}
Let $A$ be the set that is computed by the call LRR$(Q,k,0)$. Then every root of $Q$ in $\L_q[X]$ is contained in $A$.
\end{lemma}
\begin{proof}
Let 
$$
f(X)=f_0X+f_1X^q+\dots+f_{k-1}X^{q^{k-1}}
$$
be a root of $Q(X,Y)$ in $\L_q[X]$. Define $P_i$, $T_i$ and $H_i$ as in \Lref{lemma28}. We prove by induction on $i$ for $i=0,1,\dots,k-1$ that there is a recursion descent in LRR such that recursion level $i$ is called with the parameters $(P_i,k,i)$.

The base of induction is $i=0$ which is obvious. Suppose that it is true for some $i$. Then by \Lref{lemma28}, $H_i(0,f_i)=0$ and therefore, $\gamma=f_i$ is one of the roots. If $i<k-1$, then for $\gamma=f_i$ the recursive call is made with parameters 
$$
(T_i(X,Y^q+f_iX),k,\lambda+1)=(P_{i+1}(X,Y),k,i+1)
$$
If $i=k-1$, then by \Lref{lemma28}, $P_{k-1}(X,f_{k-1})=0$ which means that $f$ is inserted into $A$. 
\end{proof}
\begin{theorem}
\label{thm3}
The LRR algorithm is correct i.e. for any polynomial $Q$ as defined in \eq{Q-form-LRR}, the call LRR$(Q,k,0)$ computes a set $A$ which consists of all roots of $Q$ in $\L_q\left[X\right]$.
\end{theorem}
The proof follows from \Lref{lemma30} and \Lref{lemma29}.

%-------------------------- References-------------------------------------

\bibliographystyle{IEEEtran}

\end{document}